\pdfoutput=1
\documentclass[11pt]{article}
\usepackage{fullpage}
\usepackage{amsmath,amssymb,amsthm,mathtools,enumerate}
\usepackage[mathscr]{eucal}
\usepackage{comment}
\usepackage[usenames]{color}
\usepackage[pdftex]{graphicx} %use for pdflatex

\newcommand{\R}{\mathbb{R}}

\newcommand{\sA}{\mathscr{A}}
\newcommand{\sD}{\mathscr{D}}
\newcommand{\sDa}{\sD_{\mathrm{a}}}

\newcommand{\sX}{\mathscr{X}}

\newcommand{\bfbeta}{\mbox{\boldmath$\beta$}}
\newcommand{\bfm}{\mbox{\boldmath$m$}}
\newcommand{\sbfbeta}{\mbox{\scriptsize \boldmath$\beta$}}
\newcommand{\sbfm}{\mbox{\scriptsize \boldmath$m$}}

\newcommand{\qdel}{\delta}
\newcommand{\sigmin}{\sigma_{\mathrm{min}}}

\newtheorem{theorem}{Theorem}[section]
\newtheorem{lemma}[theorem]{Lemma}
\newtheorem{proposition}[theorem]{Proposition}

\newtheorem{remark}[theorem]{Remark}

\newtheorem{corollary}[theorem]{Corollary}

\title{Distributed noise-shaping quantization: I. 
Beta duals of finite frames and near-optimal quantization of random measurements}
\author{Evan Chou\footnote{Current affiliation: Google New York. email: \tt{chou@cims.nyu.edu}}\, 
and 
C.~Sinan G\"unt\"urk\footnote{Address: Courant Institute, 251 Mercer Street, New York, NY 10012. email: \tt{gunturk@cims.nyu.edu}}
 \\ Courant Institute, NYU}

\date{submitted: May 18, 2014, last revised: Feb 24, 2016}
\begin{document}

\maketitle

\begin{abstract}
This paper introduces a new algorithm for the so-called ``Analysis Problem'' in 
quantization of finite frame representations which provides a near-optimal solution in the case of random measurements. The main contributions include 
the development of a general quantization framework called {\em distributed noise-shaping}, and in particular, {\em beta duals} of frames, as well as the performance analysis of beta duals in both deterministic and probabilistic settings. It is shown that for random frames, using beta duals results in near-optimally accurate reconstructions with respect to both the frame redundancy and the number of levels that the frame coefficients are quantized at. More specifically, for any frame $E$ of $m$ vectors in $\R^k$ except possibly from a subset of Gaussian measure exponentially small in $m$ and for any number $L \geq 2$ of quantization levels per measurement to be used to encode the unit ball in $\R^k$, there is an algorithmic quantization scheme and a dual frame together which guarantee a reconstruction error of at most $\sqrt{k}L^{-(1-\eta)m/k}$, where $\eta$ can be arbitrarily small for sufficiently large problems. Additional features of the proposed algorithm include low computational cost and parallel implementability. 
\end{abstract}


\section{Introduction}

Let $\sX$ denote a bounded subset of $\R^k$, such as the unit ball $B^k_2$ with respect to the $\ell^2$-norm,
and $\sA$ denote a quantization alphabet (a finite subset of $\R$), such as an arithmetic
progression. Given a finite frame, i.e., a list of $m$ vectors 
$F:=\{f_1,\dots,f_m\}$ that span $\R^k$, 
we consider 
\begin{equation}\label{synthesis-dist}
 \sD_{\mathrm{s}}(F,\sA,\sX) := \sup_{x \in \sX} \inf_{q \in \sA^m} \left \| x - \sum_{i=1}^m q_i f_i \right \|_2
\end{equation}
which represents the smallest error of approximation that can be attained uniformly over $\sX$ by considering linear combinations of elements of $F$ with coefficients in $\sA$. We will call this quantity the {\em synthesis distortion} associated with $(F,\sA,\sX)$ (hence the subscript `s'). Understanding the nature of 
the synthesis distortion and finding quantization algorithms (i.e., maps
$x \in \sX \mapsto q \in \sA^m$) that achieve or even 
approximate this quantity is generally an open problem, usually referred to as
the ``Synthesis Problem'' in the theory of quantization of frame representations
\cite{powell2013quantization}. 

The synthesis distortion can be seen as a generalized notion of {\em linear discrepancy} \cite{MathAtoD, molino2012approximation}. Recall (e.g. \cite{chazelle2002discrepancy,matousek_discrepancy}) that the linear discrepancy of a $k \times m$ matrix $\Phi$ is defined as
\begin{equation}\label{defdisc}
 \mathrm{lindisc}(\Phi) := \sup_{y \in [-1,1]^m} \inf_{q \in \{-1,1\}^m} \|\Phi(y-q)\|_\infty. 
\end{equation}
We may define $\mathrm{lindisc}_p(\Phi)$ by replacing the $\infty$-norm in \eqref{defdisc} with any $p$-norm, $1 \leq p \leq \infty$. We would then have 
\begin{equation*}
\mathrm{lindisc}_2(F) =  \sD_{\mathrm{s}}(F,\{-1,1\},F([-1,1]^m)), 
\end{equation*}
where $F$ also stands for the matrix whose columns are $(f_i)_1^m$.

Consider any other frame $E := \{e_1,\dots,e_m\}$ in $\R^k$ that is dual to $F$, i.e.
\begin{equation}\label{dual-frame}
\sum_{i=1}^m \langle x,e_i \rangle f_i = x \mbox{ for all } x \in \R^k.
\end{equation}
We will write $E$ for the matrix whose
{\em rows} are $(e^\top_i)_1^m$. (Depending on the role it assumes, 
$E$ will be called the analysis frame or the sampling operator, and similarly
$F$ will be called the synthesis frame or the reconstruction operator.)
It is then evident that \eqref{dual-frame} is equivalent to $FE = I$ 
and also that duality of frames is a symmetric relation. 
We define the {\em analysis distortion} associated with $(E,\sA,\sX)$ by
\begin{equation}\label{analysis-dist}
 \sDa(E,\sA,\sX) := \inf \{\sD_{\mathrm{s}}(F,\sA,\sX) : FE= I \}.
\end{equation}

The analysis distortion has the following interpretation: Given an analysis frame $E$,
the encoder chooses a quantization map $Q:\R^m \to \sA^m$ and the decoder chooses 
a synthesis frame $F$ that is dual to $E$. The encoder maps each $x \in \sX$ (or equivalently, 
the measurement vector $y := Ex$) to its quantized version $q := Q(Ex)$, and using $q$
the decoder produces $Fq$ as its approximation to $x$. 
The optimal choice of $Q$ and $F$ (dual to $E$) yields the analysis distortion. We will refer to this optimization problem as the ``Analysis Problem''. (Also see \cite{powell2013quantization}.)

There are some remarks to be made regarding the definition of the analysis distortion and
its above interpretation: 

First, the assumption that 
the reconstruction must be ``linear'' is made for practical reasons. Non-linear methods (such as {\em consistent reconstruction} \cite{thao1996lower,goyal1998quantized}) have also been considered 
in the literature for the above type of encoders, but these methods generally come with higher
computational cost. 

Second, the assumption that $F$ be dual to $E$ is made so that 
the reconstruction is guaranteed to be exact in the high-resolution quantization limit, or simply when there is no quantization. It should be noted that 
this assumption is also meaningful in order to make sense of how the choice of $E$ might influence the reconstruction accuracy: without any assumptions on $F$ or $Q$, $F$ can be chosen optimally (in the sense
of minimizing the synthesis distortion for $\sX$) and $Q$ can be chosen to produce the corresponding
optimally quantized coefficients, thereby erasing the influence of $E$.

To make these points more precise, let us define, for any pair of frames $(E,F)$ (not necessarily dual to each other) and any quantization map $Q:\R^m \to \sA^m$,
\begin{equation} \label{DEQFX}
 \sD(E,Q,F,\sX) := \sup_{x \in \sX} \| x - F Q (Ex) \|_2.
\end{equation}
With this definition, we trivially obtain $\sD(E,Q,F,\sX) \geq \sD_{\mathrm{s}}(F,\sA,\sX)$ so that 
\begin{equation} \label{infDEQFX}
\inf_Q ~ \sD(E,Q,F,\sX) \geq \sD_{\mathrm{s}}(F,\sA,\sX).
\end{equation}
In fact, \eqref{infDEQFX} is an equality. To see this, note that the infimum in \eqref{synthesis-dist} is achieved for every $x$, so there is an optimal 
$q^*(x) \in \sA^m$ that yields
\begin{equation*}
\sD_{\mathrm{s}}(F,\sA,\sX) = \sup_{x\in \sX} \| x - F q^*(x) \|_2.
\end{equation*}
Since $E$ is injective, a map $Q^*:\R^m \to \sA^m$ can now be found such that $Q^*(Ex) = q^*(x)$, which then yields
\begin{equation*}
\sD_{\mathrm{s}}(F,\sA,\sX) = \sD(E,Q^*,F,\sX) =  \inf_Q  ~\sD(E,Q,F,\sX). 
\end{equation*}
This relation is valid for all $E$ and $F$. Fixing $E$ and minimizing over its dual frames $F$ now 
results in the interpretation of the analysis distortion given in the previous paragraph:
\begin{equation} \label{DaEAXeq}
\sD_{\mathrm{a}}(E,\sA,\sX) = \inf_{(Q,F) : FE = I}  ~\sD(E,Q,F,\sX).
\end{equation}

Finally, it should be noted that in the setup of this paper we do not allow the decoder to pick a different linear map $F_q$ depending on the quantized vector $q$ that it receives. While such an adaptive method could be considered in practice, it would generally be classified as a non-linear decoder. 

The distortion performance of any given encoder-decoder scheme (adaptive or not) that 
uses at most $N$ codewords for $\sX$ is constrained by the
entropic lower bound $\varepsilon(N,\sX)$ which is defined to
be the smallest $\varepsilon>0$ for which there exists an $\varepsilon$-net\footnote{As is common, our convention for $\varepsilon$-nets is synonymous to $\varepsilon$-coverings: $\mathcal N$ is an $\varepsilon$-net for $\sX$ if for all $x \in \sX$ there exists $y \in \mathcal{N}$ with $\|x-y\|_2\leq \varepsilon$.} for $\sX$ of cardinality $N$.
Then we have
\begin{equation}\label{dist-relations-1}
\sD(E,Q,F,\sX) \geq \sD_{\mathrm{s}}(F,\sA,\sX)
\geq \sD_{\mathrm{a}}(E,\sA,\sX) \geq \varepsilon(|\sA|^m,\sX)
\end{equation}
where the second inequality assumes that $FE=I$.
The entropic lower bound depends on $\sA$ only through its cardinality, so it will not necessarily be
useful without a suitable matching between $\sA$ and $\sX$. For example, if 
the convex hull of $F(\sA^m)$ is small compared to $\sX$, then  
$\sD_{\mathrm{s}}(F,\sA,\sX)$ will be artificially large. 
For the rest of this paper, we will fix $\sX = B^k_2$. With slight abuse of notation, define
\begin{equation} \label{DsFL}
\sD_{\mathrm{s}}(F,L) := \inf_{\sA : |\sA| = L} \sD_{\mathrm{s}}(F,\sA,B^k_2),
\end{equation}
and
\begin{equation} \label{DaEL}
\sD_{\mathrm{a}}(E,L) := \inf_{\sA : |\sA| = L} \sD_{\mathrm{a}}(E,\sA,B^k_2).
\end{equation}

Since the cardinality $N$ of any $\varepsilon$-net for $B^k_2$ must satisfy the volumetric constraint $N \mathrm{vol}(\varepsilon B^k_2) \geq  \mathrm{vol}(B^k_2)$, we have $\varepsilon(N,B^k_2) \geq N^{-1/k}$. Hence, the relations in \eqref{dist-relations-1} and the definitions in \eqref{DsFL} and \eqref{DaEL} yield, with $N=L^m$, 
\begin{equation*}
\sD_{\mathrm{s}}(F,L) \geq \sDa(E,L) \geq L^{-m/k}. 
\end{equation*}

One of the main results of this paper is a matching upper bound for the analysis distortion for Gaussian random frames. We call $E$ a {\em standard Gaussian random frame} if $(E_{i,j})$ are independent standard Gaussian random variables.

\begin{theorem} \label{T:MAIN}
There exist absolute positive constants $c_1$ and $c_2$ such that for any scalar
$\eta \in (0,1)$, if $k \geq \frac{4}{\eta}$ and $\frac{m}{k} \geq \frac{c_1}{\eta^2} \log \frac{m}{k}$, then
with probability at least $1 - \exp(- c_2 \eta^2 m)$,
a standard Gaussian random frame $E$ of $m$ vectors in $\R^k$ satisfies
\begin{equation}\label{main-thm-bound}
 \sDa(E,L) \leq \sqrt{k} L^{-(1-\eta)m/k}.
\end{equation}
for all $L \geq 2$.
\end{theorem}

The proof of this theorem is constructive and algorithmic. Given $\eta$, we will, for each $L$,
provide a quantization alphabet $\sA$ which is an arithmetic progression of size $L$, and an algorithmic
quantization map $Q:\R^m \to \sA^m$ such that with high probability on $E$, 
there exists a certain dual
frame $F$, determined by $E$ and $Q$ and computable by a fast, low-complexity algorithm,
which yields $\sD(E,Q,F,B^k_2) \lesssim (m/k)^{3/2} k^{-1}  L^{-(1-\eta)m/k}$. The chosen $Q$ is a variation on vector-valued $\beta$-expansions and $F$ is an alternative dual of $E$ determined by $Q$. Suitably restricting the range of $m$ and $k$ and
summing over the failure probability for each $L$ will then allow us to guarantee the distortion bound \eqref{main-thm-bound} uniformly over $L$.

Quantization of finite frames by sigma-delta modulation goes back to \cite{benedetto2006sigma,BPY2}.
Alternative duals for frames were introduced in the context of sigma-delta quantization in \cite{LPY,alternative2010,blum2010sobolev}, and more recently have been effectively used for Gaussian random
frames in \cite{gunturk2010sobolev} and more generally for sub-Gaussian random frames in \cite{krahmer2013subGaussian}. For these frames,
an $r$-th order sigma-delta quantization scheme coupled with an associated Sobolev dual can achieve the bound
$\sDa(E,L) \lesssim L^{-1} (cr)^{2r} (m/k)^{-r}$ with high
probability for some constant $c$. 
Optimizing this bound for $r$ yields $\sDa(E,L) \lesssim L^{-1}e^{-\tilde c\sqrt{m/k}}$ for another  
constant $\tilde c$ which is sub-optimal in terms of its dependence on both the alphabet size $L$ and the oversampling factor $m/k$ and $L$. In contrast,
our result achieves near-optimal dependence on both parameters: Not only is the dependence on $m/k$ 
exponential, but the rate can also be made arbitrarily close to its maximum value $\log_2 L$.
To the best of our knowledge, this is the first result of this kind.\footnote{Similar exponential error bounds that have been obtained previously in the case of conventional sigma-delta modulation or for other quantization schemes are not compatible with the results of this paper: The results of \cite{gunturk2003one} and \cite{deift2011optimal} are for a fixed frame-like system, but using a different norm in infinite dimensions, and the dependence on $L$ is unavailable. The results in \cite{derpich2008}, obtained in yet another functional setting for sigma-delta modulation, come close to being optimal, however these results were obtained under modeling assumptions on the quantization noise and circuit stability. The exponential near-entropic error decay in the bitrate obtained in \cite{IwenSaab} combine sigma-delta modulation with further (lossy) bit encoding. Finally, the exponential error decay reported in \cite{BFNPW14} is obtained with adaptive hyperplane partitions and does not correspond to linear reconstruction.}

The paper is organized as follows: In Section \ref{section-distributed-quantization}, we review the notion of general noise-shaping quantizers and their stability as well as alternative duals of frames based on noise shaping. In this section we introduce a number of new notions, namely, $V$-duals and $V$-condensation of frames and the concept of {\em distributed noise shaping}. In Section \ref{section-beta-duals}, we introduce the special case of {\em beta duals} and derive a general bound for the analysis distortion of an arbitrary frame using beta duals. As an example, we show in this section that beta duals are optimal for harmonic semicircle frames. In Section \ref{gaussian_frames}, we bound the analysis distortion of Gaussian frames, prove Theorem \ref{T:MAIN} as well as some variations of it, and also discuss the extension of our results to sub-Gaussian frames.

This paper is the first in a series of papers (in progress) on distributed noise shaping. The follow-up papers will address extensions of the theory developed in this paper for frames to more general sampling scenarios, including compressive sampling and infinite dimensional systems.

\section{Quantization via distributed noise shaping} \label{section-distributed-quantization}

\subsection{Noise-shaping quantizers} \label{noise-shaping}

Let $\sA$ be a quantization alphabet and $J$ be a compact interval in $\R$.
Let $h = (h_j)_{j\geq0}$ be a given sequence, finite or infinite, and $h_0=1$.
By a noise-shaping quantizer with the transfer sequence $h$, we  
mean any map $Q:J^m \to \sA^m$ where $q := Q(y)$ satisfies
\begin{equation*}
y - q = h*u
\end{equation*}
and $\|u\|_\infty \leq C$ for some constant $C$ which is independent of $m$.
Here $h*u$ refers to the convolution of $h$ and $u$ defined by
$(h*u)_n := \sum h_j u_{n-j}$ where it is assumed that $u_n := 0$ for $n\leq 0$.

Noise-shaping refers to the fact that
the ``quantization noise'' $y-q$ cannot be arbitrary.
Even though the operator
$H:u \mapsto h*u$ is invertible on $\R^m$; the requirement that 
$\|u\|_\infty$ must be controlled uniformly in $m$ imposes restrictions on what $q$ can be
for a given $y$.

The above formulation amounts to saying that $H$ is an $m\times m$ 
lower triangular Toeplitz matrix defined by $h$, in particular with unit diagonal.
Let us relax the notion of a noise-shaping quantizer and assume that $H$ is any given $m\times m$
lower triangular matrix with unit diagonal. We will refer to $H$ as the noise transfer operator and 
keep the noise-shaping relation
\begin{equation} \label{y-q-H-u}
  y - q = Hu,
\end{equation}
but since $H$ is more general and $m$ is therefore static, we can no longer insist that 
$\|u\|_\infty$ is controlled in the same manner. While we will mostly stick to convolution-based
noise-shaping as defined at the beginning, the above relaxation is conceptually and algebraically 
simpler.

Noise-shaping quantizers do not exist unconditionally. However, 
under certain suitable assumptions on $H$ and $\sA$, they exist and 
can be implemented via recursive algorithms. The simplest
is the greedy quantizer whose general formulation is given below:

\begin{proposition}\label{P:greedy}
Let $\sA := \sA_{L,\delta}$ denote the arithmetic progression in $\R$ which is of length $L$, spacing $2\delta$, and symmetric about $0$.
Assume that $H = I - \tilde H$, where $\tilde H$ is strictly lower triangular, and $\mu \geq 0$ such that 
$\|\tilde H\|_{\infty \to \infty} + \mu/\delta \leq L$. Then 
there exists a quantization map $Q_{H,\sA}:[-\mu,\mu]^m \to \sA^m$ implementable by a recursive
algorithm such that $q:=Q_{H,\sA}(y)$ satisfies \eqref{y-q-H-u}
and $\|u\|_\infty \leq \qdel$.
\end{proposition}
\begin{proof}
Let $\mathrm{round}_\sA:\R\to \sA$ be any rounding function satisfying $|w - \mathrm{round}_\sA(w)| \leq |w-a|$ for all $a\in\sA$ and $w \in \R$.
Since $\sA = \{(-L+2l-1)\delta : 1\leq l \leq L \}$, we have $|w - \mathrm{round}_\sA(w)| \leq \delta$ for all $|w| \leq L\delta$.
For $n=1,\dots,m$, we set 
\begin{equation*}
q_n := \mathrm{round}_\sA\left (y_n + \sum_{j=1}^{n-1} \tilde H_{n,n-j} u_{n-j} \right) 
\end{equation*}
and
\begin{equation*}
u_n := y_n + \sum_{j=1}^{n-1} \tilde H_{n,n-j} u_{n-j} -q_n .
\end{equation*}
The proof now follows by induction. Setting $w=y_1$, we obtain $|u_1| \leq \delta$.
Assuming that $|u_j|\leq \delta$ for $j=1,\dots,n-1$, we obtain
\begin{equation*}
\left |y_n + \sum_{j=1}^{n-1} \tilde H_{n,n-j} u_{n-j} \right| 
\leq \mu + \delta \|\tilde H\|_{\infty \to \infty} \leq L\delta.
\end{equation*}
Hence $|u_n| \leq \delta$.
\end{proof}
\begin{remark}\label{greedyconv}
{\rm 
In the special case $Hu = h*u$ where $h_0 = 1$, note that $\|\tilde H\|_{\infty \to \infty} = \|h\|_1 - 1$. This basic fact will be invoked in Section \ref{section-beta-duals}.
}
\end{remark}

\subsection{Alternative duals of frames for noise shaping} \label{V-duals}
Let $E$ be a frame and $y=Ex$ be the frame measurements of a given signal $x$. Assume that
we quantize $y$ using a noise-shaping quantizer with transfer operator $H$.
Any left-inverse (dual) $F$ of $E$ gives
\begin{equation*}
x - Fq = F(y - q) = FHu. 
\end{equation*}

As explained in the introduction, the choice of $F$ is to be made based upon $E$ and the quantizer 
$Q_{H,\sA}$, but not $u$ which is unavailable to the decoder.
We are interested in bounding $\|x - Fq\|_2$, but the best a priori bound for $u$ is 
in the $\infty$-norm. This suggests that we employ the bound
\begin{equation}\label{error_bound_FHu}
 \|x - Fq \|_2 \leq \|FH\|_{\infty \to 2} \|u \|_\infty.
\end{equation}
With this bound, the natural objective would be to minimize $\|FH\|_{\infty \to 2}$ over all duals of $E$.
However, an analytical solution for this problem is not available to us. If instead we minimize 
$\|FH\|_\mathrm{2 \to 2}$, then the solution is known explicitly, and given by 
$F = (H^{-1}E)^\dagger H^{-1}$,
where $A^\dagger := (A^\ast A)^{-1} A^\ast$ is the pseudoinverse \cite{blum2010sobolev,gunturk2010sobolev}. 

More generally, for any $p \times m$ matrix $V$ such that $VE$ is also a frame for $\R^k$ 
(i.e. of rank $k$),  
\begin{equation*}
F_V := (VE)^\dagger V
\end{equation*}
is a dual of $E$, which we will call the {\em $V$-dual} of $E$. (For the case $V = H^{-1}$, 
the resulting dual was called the $H$-dual in \cite{gunturk2010sobolev} and denoted by $F_H$.
With our more general definition, this dual more accurately becomes the $H^{-1}$-dual 
and denoted by $F_{H^{-1}}$.) 
The particular case $V = D^{-r}$, called the ``Sobolev dual'', 
was introduced in \cite{blum2010sobolev} and studied for random $E$ in \cite{gunturk2010sobolev}.

With a $V$-dual, we have $F_V H = (VE)^\dagger VH$. Given that $\| (VE)^\dagger \|_{2\to 2}
 = 1/\sigmin( VE)$, we employ the initial bound
\begin{equation*}
\|F_V H\|_{\infty \to 2} \leq \frac{\| VH\|_{\infty \to 2} }{\sigmin( VE) }. 
\end{equation*}
Note that this bound is still valid when $VE$ fails to be a frame, since then we have $\sigmin(VE)=0$.

Let us compare the following two routes for bounding $\| VH\|_{\infty \to 2}$:
\begin{itemize}
 \item[(i)]  $\| VH\|_{\infty \to 2} \leq \sqrt{m}\| VH\|_{2 \to 2} $. Since
$\| VH  \|_\mathrm{2 \to 2}/\sigmin( VE) \geq \|F_V H\|_{2\to2} \geq 1/\sigmin(H^{-1}E)$, it follows
that the minimum value of $\| VH  \|_\mathrm{2 \to 2}/\sigmin( VE)$ is attained when $V = H^{-1}$. 
Hence, this route does not offer a new bound.
 \item[(ii)]  $\| VH\|_{\infty \to 2} \leq \sqrt{p}\| VH\|_{\infty \to \infty} $. The minimization of
$\sqrt{p} \| VH  \|_\mathrm{\infty \to \infty}/\sigmin( VE)$ opens up 
the possibility of a smaller error bound because $V = H^{-1}$ readily achieves the same bound
as in (i).\footnote{The following example shows that the minimum value can be 
strictly smaller: Let 
$H = \left[\begin{matrix}1&0\\-1&1\end{matrix}\right]$, $E =\left[\begin{matrix}1\\1\end{matrix}\right] $
for which $\sqrt{m}/\sigmin(H^{-1}E) = \sqrt{2/5}$. Meanwhile, 
$V = \left[\begin{matrix}1&1\end{matrix}\right]$ yields
$\sqrt{p}\|VH\|_{\infty \to \infty}/\sigmin(VE) = 1/2$.}
Let $v_1,\dots,v_p$ denote the rows of $V$. Note that 
\begin{equation*}
\| VH  \|_{\infty \to \infty} =  \max_{1 \leq n \leq p} \| v_n H \|_1.
\end{equation*}
Hence an effective strategy for designing $V$ is based on choosing $v_1,\dots,v_p$ near (but not
necessarily equal to) the 
bottom $p$ left singular vectors of $H$. We will exploit this principle in the next section.
\end{itemize}

Note that different choices of $V$ can yield the 
same $V$-dual. In particular, given any $V$, the $k \times m$ matrix 
$ \tilde V:= (W^*W)^{-1/2}W^*V $ yields $F_{\tilde V} = F_V$, where $W := VE$. 
This $\tilde V$ is not particularly useful for computational purposes, but it shows that 
$p$ can be chosen as small as $k$ without sacrificing performance.
In general, when $p < m$, we will call $VE$ the {\em $V$-condensation} of $E$.

\subsection{Distributed noise shaping} \label{distributed}

By a {\em distributed} noise-shaping quantization scheme, we will mean a noise-shaping quantizer defined by
a block-diagonal noise transfer matrix $H$, and for any given analysis frame $E$, an associated $V$-dual where
$V$ is also block-diagonal with matching sized blocks. More specifically, we set
\begin{equation*}
H=\left[\begin{matrix} H_1  & & \\ & \ddots & \\ & & H_l \end{matrix}\right]
\quad\text{and}\quad
V=\left[\begin{matrix} V_1  & & \\ & \ddots & \\ & & V_l \end{matrix}\right]
\end{equation*}
where $H_i\in\R^{m_i\times m_i}$ and $V_i\in\R^{p_i\times m_i}$ with $\sum_i m_i=m$ and $\sum_i p_i=p$.
We further decompose $E$, $y$, $q$ and $u$ as
\begin{equation*}
E = \left[\begin{matrix}E_1 \\ \vdots \\ E_l \end{matrix}\right],\quad
y = \left[\begin{matrix}y_1 \\ \vdots \\ y_l \end{matrix}\right],\quad
q = \left[\begin{matrix}q_1 \\ \vdots \\ q_l \end{matrix}\right]
\quad\text{and}\quad
u = \left[\begin{matrix}u_1 \\ \vdots \\ u_l \end{matrix}\right],
\end{equation*}
where $E_i \in \R^{m_i \times k}$, and $y_i$, $q_i$, $u_i \in \R^{m_i}$, $i \in [l]:=\{1,\dots,l\}$.
Note that the $E_i$ may or may not be individual frames in $\R^k$. With this notation, we have
$l$ individual quantizers that run in parallel:
\begin{equation}\label{y-q-H-u-i}
E_i x - q_i = H_i u_i,\qquad i \in [l].
\end{equation}
Let $v_{i,j}$ denote the $j$th row of $V_i$, $i \in [l]$, $j\in[p_i]$. Then we have
\begin{equation*}
\| VH  \|_{\infty \to \infty} =  \max_{i \in [l]} \max_{j \in [p_i]} \| v_{i,j} H_i \|_1.
\end{equation*}
With this expression it follows from the matrix-norm bound
\begin{equation*}
\|F_V H\|_{\infty \to 2} \leq \frac{\sqrt{p} \| VH\|_{\infty \to \infty}}{\sigmin( VE)}
\end{equation*}
derived in Section \ref{V-duals} (via route (ii)) for $V$-duals and the general noise-shaping error bound \eqref{error_bound_FHu} that 
\begin{equation*}
\|x - F_V q\|_2 \leq \frac{\sqrt{p} \|u \|_\infty}{\sigmin(VE)} 
 \max_{i \in [l]} \max_{j \in [p_i]} \|v_{i,j} H_i \|_1.
\end{equation*}

\begin{remark}
{\rm
For conceptual clarity, we may assume that $l \not= 1$ so that distributed noise shaping is genuinely a special case of noise shaping as described in Section \ref{V-duals}. However, we will not need to turn this into an explicit assumption because we will only consider $l \geq k$, and most often, $l > k$.
}
\end{remark}

\begin{remark}
{\rm 
Our formalism for noise-shaping quantization and alternate-dual reconstruction assumes
that the frame vectors and the measurements are processed in the given (specified) order. 
For distributed noise shaping, we assume that these vectors are collected into groups according to the same ordering. 
}
\end{remark}

\begin{remark}
{\rm 
It is possible to view distributed noise-shaping quantization as a generalization of parallel sigma-delta modulation introduced in the analog-to-digital conversion literature \cite{GJ}.
}
\end{remark}

\section{Beta duals}\label{section-beta-duals}

We now propose a specific distributed noise-shaping quantization scheme and corresponding
alternative duals for reconstruction which will form the basis for the main results of this paper.

For any $\beta \geq 1$, let $h^\beta$ be the (length-$2$)
sequence given by $h^\beta_0 = 1$ and $h^\beta_1 = -\beta$. This paper is about the implications of setting $\beta > 1$, but the case $\beta = 1$ is permissible, too.

Given $k \leq l \leq m$, let $\bfm:=(m_1,\dots,m_l)$ and $\bfbeta:=(\beta_1,\dots,\beta_l)$ be such that
$\sum m_i = m$ and $\beta_i \geq 1$, $i \in [l]$. For each $i$, we set $H_i$
equal to the convolution operator $u \mapsto h^{\beta_i}*u$, $u \in \R^{m_i}$, as defined in
Section \ref{noise-shaping}.
Note that each $H_i$ is an $m_i \times m_i$ bidiagonal matrix with 
diagonal entries equal to $1$ and the subdiagonal entries equal to $-\beta_i$. As in Remark \ref{greedyconv}, we have $\|\tilde H_i \|_{\infty \to \infty} = \beta_i$.

We set $p_i = 1$ and
$V_i =  v_{i,1} := [\beta_i^{-1} ~~ \beta_i^{-2}~\cdots~~\beta_i^{-m_i}]$, $i \in [l]$.
Let $V =: V_{\sbfbeta,\sbfm}$ and $H$ be as in Section \ref{distributed}. Note that $p=l$. It follows that 
$V_i H_i = [0~\cdots~0~~\beta_i^{-m_i}]$, 
$\|V_i H_i\|_{\infty \to \infty} = \beta_i^{-m_i}$, and therefore 
$\|VH\|_{\infty \to \infty} = \beta_*^{-m_*}$, where $\beta_* := \min_i \beta_i$ and $m_* := \min_i m_i$.
With the setup and the notation of Section \ref{distributed},
it follows that 
\begin{equation} 
\|x - F_{V_{\sbfbeta,\sbfm}} q\|_2 \leq 
\frac{\sqrt{l} \|u\|_\infty}{\sigmin(V_{\sbfbeta,\sbfm}E)} \beta_*^{-m_*}.
\label{E:bdest}
\end{equation} 
Also note that we have $\|\tilde H \|_{\infty \to \infty} = \max_i \beta_i$. 

Note that this particular choice of $V_i = v_{i,1}$ is derived from 
$\beta$-expansions of real numbers. In fact, solving the 
difference equations \eqref{y-q-H-u-i} for $q_i$ and recovering with
$v_{i,1}$ gives a truncated $\beta_i$-expansion of $v_{i,1}E_ix$. (Here we assume the $\beta_i > 1$.) For theory and applications of $\beta$-expansions, see \cite{Parry, Dajani, DDGV}. 

We will refer to $V_{\sbfbeta,\sbfm}E$ as the {\em $(\bfbeta,\bfm)$-condensation} of $E$, and 
the corresponding $V_{\sbfbeta,\sbfm}$-dual $F_{V_{\sbfbeta,\sbfm}}$ as the 
{\em $(\bfbeta,\bfm)$-dual} of $E$. We will use the term {\em beta dual} when implicitly referring to 
any such dual. The case $\beta_i=1$ corresponds to sigma-delta modulation within the $i$th block, but the duals considered in this paper will differ from Sobolev duals and therefore we will continue to use the term beta dual.

In a typical application all the $\beta_i$ would be set equal to some $\beta = \beta_*$ and the $m_i$
(approximately) equal to $m/l$. To denote this choice, we will replace the vector pair $(\bfbeta, \bfm)$
by the scalar triplet $(\beta, m, l)$, ignoring the exact values of the $m_i$ as 
long as $m_* \geq \lfloor m/l \rfloor$.

\subsection{Bounding the analysis distortion via beta duals}
\label{betadualsteps}

For any integer $L \geq 2$ and scalar $\mu > 0$, let 
\begin{equation*}
 S_{\mu,L}:= \{(\beta,\delta) : \beta \geq 1, ~\delta > 0, ~\beta + \mu/\delta \leq L \}.
\end{equation*}
This set corresponds to the admissible set of $\beta$ and $\delta$ values such that the greedy 
quantization rule described in Proposition \ref{P:greedy} and Remark \ref{greedyconv} using $\sA_{L,\delta}$ and a noise 
shaping operator $H$ as defined above 
(with all $\beta_i = \beta$) is stable with $\|u\|_\infty \leq \delta$ uniformly for 
all inputs bounded by $\mu$. Note that the best $\ell_\infty$-norm
bound for $y = Ex$ over all $x \in B^k_2$ is by definition equal to $\| E \|_{2 \to \infty}$. 
Combining \eqref{E:bdest} with \eqref{DEQFX}, \eqref{DaEAXeq} and \eqref{DaEL}, we 
immediately obtain the following result:
\begin{proposition} \label{analysis_dist}
Let $E$ be an $m \times k$ frame and $L \geq 2$. For any 
$\mu \geq \| E \|_{2 \to \infty}$, $(\beta, \delta) \in S_{\mu,L}$, and $k \leq l \leq m$, 
we have
\begin{equation}\label{DaEL_bound1}
 \sDa(E,L) \leq 
 \frac{\delta \sqrt{l} \beta^{-\lfloor m/l\rfloor}}{\sigmin(V_{\beta,m,l}E)}.
\end{equation}
\end{proposition}
In order to bound $\sDa(E,L)$ effectively via this proposition, we will
do the following:
\begin{enumerate}
\item[(1)] Find an upper bound $\mu$ for $\|E\|_{2\to \infty}= \displaystyle \max_{1 \leq i \leq m} \|e_i\|_2$. Note that $\mu_1 \leq \mu_2$ implies $S_{\mu_2,L} \subset S_{\mu_1,L}$, therefore the smallest value of $\mu$ would result in the largest search domain for $(\beta,\qdel)$.
\item[(2)] Find a lower bound on $\sigmin(V_{\beta,m,l}E)$ in terms of $\beta \geq 1$ and (if chosen as an optimization parameter) $l$.
\item[(3)] Choose optimal (or near-optimal) admissible values of  $(\beta,\qdel)$ and (if chosen as an optimization parameter) $l$ for the resulting effective upper bound replacing \eqref{DaEL_bound1}. 
\end{enumerate}

The following simple lemma will be useful when choosing $(\beta,\qdel)$.

\begin{lemma}\label{calculus_lemma}
For any $L \geq 2$, $\mu > 0 $, and $\alpha \geq 0$, let $\beta = L(\alpha+1)/(\alpha+2)$ and $\delta=\mu(\alpha+2)/L$. Then $(\beta,\delta) \in S_{\mu,L}$ and
\begin{equation}\label{calculus_lemma_ineq}
\delta \beta^{-\alpha} < \mu e (\alpha+1) L^{-(\alpha+1)}. 
\end{equation}
\end{lemma}
\begin{proof}
We have $\beta \geq 1$, $\delta > 0$, and $\beta + \mu/\delta = L$ so that $(\beta,\delta) \in S_{\mu,L}$. (In fact, we have $\beta > 1$ unless $\alpha = 0$ and $L = 2$.) With these values, we obtain $\delta \beta^{-\alpha} = \mu (1+\frac{1}{\alpha+1})^{\alpha+1}(\alpha+1)L^{-(\alpha+1)}$ and \eqref{calculus_lemma_ineq} follows by noting that $(1+1/t)^t < e$ for all $t > 0$.
\end{proof}

\begin{remark}{\rm 
As it is evident from Proposition \ref{analysis_dist}, the above lemma will be used in the regime $\alpha \approx m/k$. As the number of measurements $m$ is increased, the value of $\delta$ chosen in Lemma \ref{calculus_lemma} would increase linearly with $m$, which may not be feasible in a practical circuit implementation due to dynamic range (or power) limitations. Note that the increase in $\delta$ is actually caused by the fact that the chosen value of $\beta$ approaches $L$ as $\alpha$ is increased. These values were chosen in Lemma \ref{calculus_lemma} to (near)-optimize $\delta \beta^{-\alpha}$ subject to the constraint $(\beta,\delta) \in S_{\mu,L}$. If instead a more modest, fixed value of $\beta$, such as $\beta = 2L/3$, is employed, then it suffices to set $\delta = 3\mu/L$, which is only weakly dependent on $m$ (possibly through $\mu$, depending on the nature of measurements). 
This sub-optimal choice would increase the analysis distortion bound in \eqref{calculus_lemma_ineq} roughly by a factor of $(3/2)^\alpha$ which may still be affordable in practice due to the presence of the dominating term $L^{-\alpha}$.
}
\end{remark}

As an example of the above procedure for bounding the analysis distortion of frames, we will discuss in the next subsection a specific family of (deterministic) frames, namely the harmonic semicircle frames in $\R^2$. In Section 4, we will reformulate this procedure for the setting of random frames.

\subsection{Example: the analysis distortion for the harmonic semicircle frames}

Harmonic semicircle frames  are obtained from harmonic frames by removing some of their symmetry (see \cite{framepath2007} for the general
definition). In 
$\R^2$, the harmonic semicircle frame of size $m$ is defined by the $m\times 2$ matrix 
$E:=E_{\mathrm{hsc},m}$ given by
\begin{equation*}
E_{\mathrm{hsc},m}:= \left[\begin{matrix} 
\cos \pi/m & \sin \pi/m \\
\cos 2\pi/m & \sin 2\pi/m \\
\vdots & \vdots \\
\cos \pi & \sin \pi 
\end{matrix}\right]
\end{equation*}

In Figure \ref{fig:bdukcondtwo} we illustrate the harmonic semicircle frame of size $m=12$ with its $(\beta,m,l)$-duals for $l=2$ and $l=3$. For these plots, we have used $\beta = 1.6$.

\begin{figure}[htp]
\centering
\includegraphics[scale=0.5]{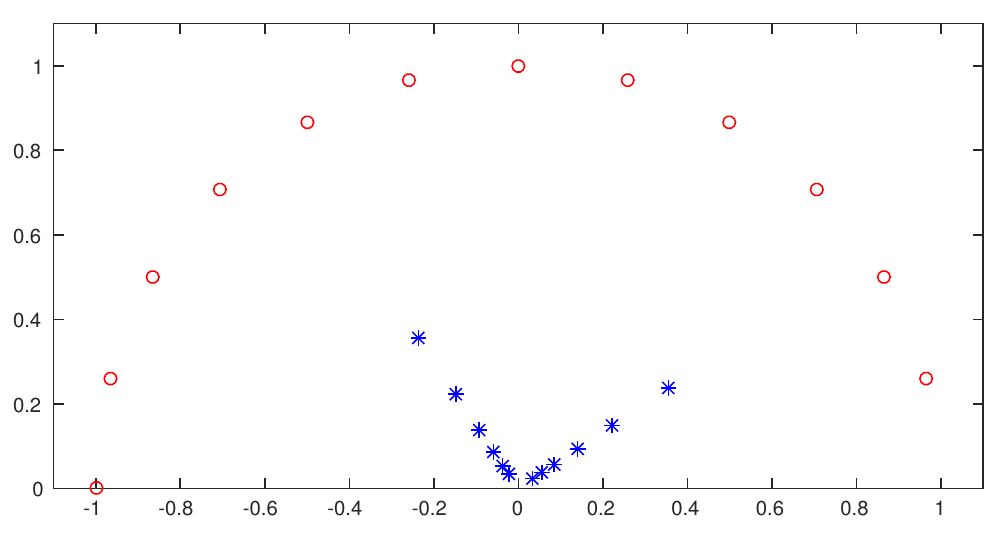}\\
\includegraphics[scale=0.5]{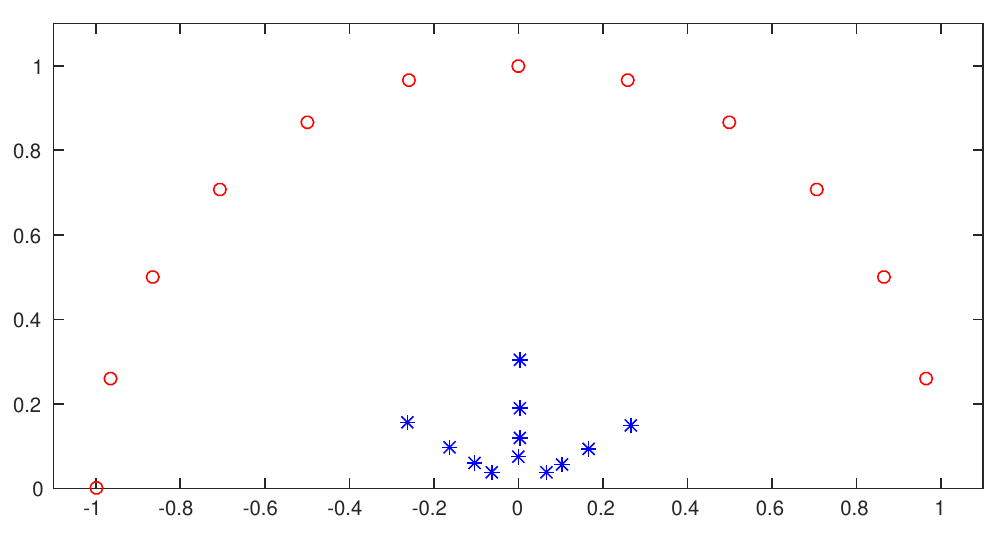}\\
\caption{Harmonic semicircle frame (circles) of size $m=12$ and two of its $(\beta,m,l)$-duals (stars) for $\beta=1.6$. Top: $l=2$, bottom: $l=3$.}
\label{fig:bdukcondtwo}
\end{figure}

To bound the analysis distortion of the harmonic semicircle frames, we start by setting $\mu=1$ since all the frame vectors are of unit norm. 

We will work with the choice $l=2$, and for simplicity, assume that $m$ is even. Let $E_1$ and $E_2$ be the two halves of $E$ as in Section \ref{distributed}.

For any $\beta > 1$, consider $V_{\beta,m,2}E$, which is given by the $2 \times 2$ matrix
\begin{equation*}
\left[\begin{matrix} v_\beta E_1 \\ v_\beta E_2 \end{matrix}\right] 
\end{equation*}
where $v_\beta := V_1 = V_2 = [\beta^{-1} \cdots \beta^{-m/2} ]$.
Note that each row of $E_2$ is a 90 degree rotation of the corresponding row of $E_1$. 
Therefore $ v_\beta E_2 $ is also a 
90 degree rotation of $ v_\beta E_1 $ which implies that  
$\sigmin (V_{\beta,m,2}E) = \|v_\beta E_1 \|_2$. An exact 
expression for $\|v_\beta E_1 \|_2$ follows easily using Euler's formula:
\begin{equation*} 
\|v_\beta E_1 \|_2
= \left | \sum_{k=1}^{m/2} \beta^{-k} e^{\pi i k/m} \right |
= \frac{(1+\beta^{-m})^{1/2}}{|\beta-e^{\pi i /m}|}.
\end{equation*}
It can now be deduced from this expression that $\|v_\beta E_1 \|_2 \geq \beta^{-1}$ for all $m$. Indeed, for $m=2$ we have equality, and for all $m\geq 3$ (which in turn implies $m\geq 4$), we have $|\beta-e^{\pi i /m}| \leq (\beta -1) + |1 - e^{\pi i /m}| \leq \beta$. Therefore Proposition \ref{analysis_dist} yields
\[ \sDa(E_{\mathrm{hsc},m},L) \leq \sqrt{2} \delta \beta^{-m/2+1} \]
which is valid for all $(\beta,\qdel) \in S_{1,L}$. With $\alpha:=m/2-1$, Lemma \ref{calculus_lemma} now implies
\[ \sDa(E_{\mathrm{hsc},m},L) < \frac{e}{\sqrt{2}} \,mL^{-m/2}. \]

\section{Analysis distortion of random frames} \label{gaussian_frames}
In this section, we assume that $E$ is an $m \times k$ random matrix with entries drawn independently from a given probability distribution. The main results we highlight will be for the standard Gaussian distribution $\mathcal{N}(0,1)$, though we will also discuss what can be achieved by the same methods for sub-Gaussian distributions. As before, $L$ stands for a given alphabet size. To bound the analysis distortion $\sDa(E,L)$, we will again follow the 3-step plan based on beta duals as outlined in section \ref{betadualsteps}, except the steps 1 and 2 are available only with probabilistic guarantees. The selection of the parameters $\beta$ and $l$ in step 3 will affect the level of probabilistic assurance with which the resulting distortion bound will hold.

Note that the entries of the $l \times k$ matrix $V_{\beta, m, l} E$ are mutually independent and the entries of the $i$th row are identically distributed according to $\mathcal{N}(0,\sigma^2_{m_i})$, where 
\begin{equation*}
\sigma^2_n := \sum_{j=1}^{n} \beta^{-2j},
\end{equation*}
and the $m_i$ satisfy $\sum m_i = m$ and $m_i \geq \lfloor m/l \rfloor$ for all $i$. Our analysis of $\sigmin(V_{\beta, m, l} E)$ would be simpler if all the $\sigma_{m_i}$ were equal, which holds if and only if $m$ is a multiple of $l$. However this condition cannot be enforced in advance. To circumvent this issue, we define for each $l$, the submatrix $\tilde E_l$ of $E$ formed by the first $\tilde m$ rows of $E$ where $\tilde m := \lfloor m/l \rfloor l$. Provided that $\tilde E_l$ is a frame for $\R^k$ (which holds almost surely if $l\geq k$), the definition of the analysis distortion $\sDa(E,\sA,\sX)$ given by \eqref{analysis-dist} implies that 
\begin{equation*} 
\sDa(E,\sA,\sX)  \leq \sDa(\tilde E_l,\sA,\sX), 
\end{equation*}
because for any $\tilde F$ such that $\tilde F \tilde E_l = I$, the zero-padded $k \times m$ matrix $F:= [\tilde F ~ 0]$ satisfies $FE = I$ and 
\begin{equation*}
\sD_{\mathrm{s}}(F,\sA,\sX) = \sD_{\mathrm{s}}(\tilde F,\sA,\sX). 
\end{equation*}
Hence we can replace \eqref{DaEL_bound1} with
\begin{equation}\label{DaEL_bound2}
 \sDa(E,L) \leq \frac{\delta \sqrt{l} \beta^{-\lfloor m/l\rfloor}}{\sigmin(V_{\beta,\tilde m,l} \tilde E_l)}.
\end{equation}
where again $(\beta,\qdel) \in S_{\mu,L}$ and $k \leq l \leq m$ are to be chosen suitably.

\subsection{Upper bound for $\|E\|_{2\to \infty}$}
Since the Gaussian distribution has unbounded support, there is no absolute bound on $\|E\|_{2\to \infty}$ that is applicable almost surely, so we employ concentration bounds which hold with high probability. For any $a>0$, consider the event
\begin{equation*}
{\mathcal M}(a):=\Big \{\|E\|_{2\to \infty} \leq 2(1+a)\sqrt{m}\Big \}.
\end{equation*}
Noting that $\|E\|_{2\to\infty} \leq \|E\|_{2\to 2} = \sigma_\mathrm{max}(E)$, we consider the largest singular value instead. The following well-known result will be suitable for this purpose:
\begin{proposition}[{\cite[Theorem II.13]{ds2001}}{\cite{ds2003}}]\label{ds2001}
Let $m \geq k$ and $E$ be an $m \times k$ random matrix whose entries are i.i.d. standard Gaussian variables. For any $t >0$,
\begin{equation*}
\mathbb{P} \left(\sigma_\mathrm{max}(E) \geq \sqrt{m}+\sqrt{k} + t  \right) < e^{-t^2/2}.
\end{equation*}
\end{proposition}

Hence we have, with $t=2a\sqrt{m}$ and noting $m \geq k$,
\begin{equation}\label{prob-Mc}  
\mathbb{P}({\mathcal M}(a)^c) < e^{-2a^2m}.
\end{equation}

\subsection{Lower bound for $\sigmin(V_{\beta,\tilde m,l}\tilde E_l)$} Again, we employ concentration bounds. For any $\gamma > 0$, we are interested in the event
\begin{equation}\label{Egamma}
 \mathcal{E}(\gamma,\beta,l) :=\left\{ \sigmin(V_{\beta,\tilde m,l}\tilde E_l) \geq \gamma \right\}.
\end{equation}
Note that the entries of the $l\times k$ matrix $V_{\beta, \tilde m, l} \tilde E_l$ are now independent and
identically distributed, with each entry having distribution $\mathcal{N}(0,\sigma^2_{\lfloor m/l \rfloor})$.
It will suffice to use standard methods and results on the concentration
of the minimum singular value of Gaussian matrices. The nature of these results differ depending on whether $l=k$ or $l > k$ and the optimal value of $l$ depends on the desired level of probabilistic assurance.

For $l=k$, we use the following result:
\begin{proposition}[{\cite[Theorem 3.1]{rudelson2010non}}, {\cite{edelman1988eigenvalues}}]\label{P:squaregaussminsv}
Let $\Omega$ be a $k \times k$ random matrix with entries drawn
independently from $\mathcal{N}(0,1)$. Then for any $\varepsilon > 0$,
\[ \mathbb{P}\left(\sigmin(\Omega) \leq \frac{\varepsilon}{\sqrt{k}}\right) \leq \varepsilon. \]
\end{proposition}

For $l>k$, there are standard estimates concentration estimates (e.g. \cite{ds2001, rudelson2009smallest, vershynin2010introduction}) which take the lower end of the spectrum $\sqrt{l}-\sqrt{k}$ as a reference point. 
In contrast, the following result concerns the probability of the smallest singular value being near zero. The proof, which is based on a standard covering technique, is given in Appendix \ref{appendix}.
\begin{theorem}\label{T:svmingaussrect}
Let $l > k$ and $\Omega$ be an $l \times k$ random matrix whose entries are drawn independently from $\mathcal{N}(0,1)$. Then for any $0 < \varepsilon < 1$,
\[ \mathbb{P}\left( \sigmin(\Omega) \leq \varepsilon\sqrt{l}/2\right) 
\leq
2 \left (10+8\sqrt{\log \varepsilon^{-1}}\right)^k  e^{l/2} \varepsilon^{l-k+1}.
 \]
\end{theorem}

Note that $\varepsilon$ needs to be sufficiently small for the above probabilistic bound to be effective, i.e., less than 1.

Proposition \ref{P:squaregaussminsv} and Theorem \ref{T:svmingaussrect} have a simple interpretation regarding beta duals of Gaussian frames. Suppose, for simplicity, that $m$ is divisible by both $k$ and $l$. We see that the $(\beta,m,l)$-condensation of a Gaussian frame $E$ in $\R^k$ for $l > k$ is a lot less likely to be near-singular than the $(\beta,m,k)$-condensation.  This comment applies verbatim to the $(\beta,m,l)$- and the $(\beta,m,k)$-duals of $E$ since the beta condensation and the beta dual of $E$ have the same frame ratios (condition numbers). This is illustrated in Figure \ref{fig:betarandplot}.

\begin{figure}[htp]
\centering
\includegraphics[height=6cm]{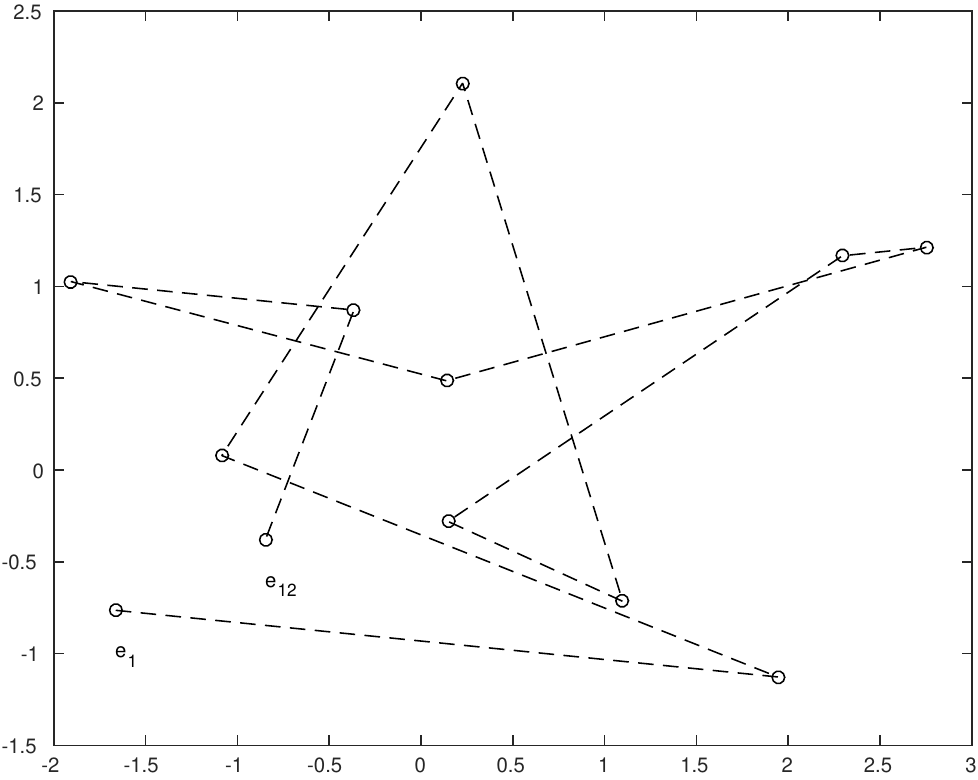}
\includegraphics[height=6cm]{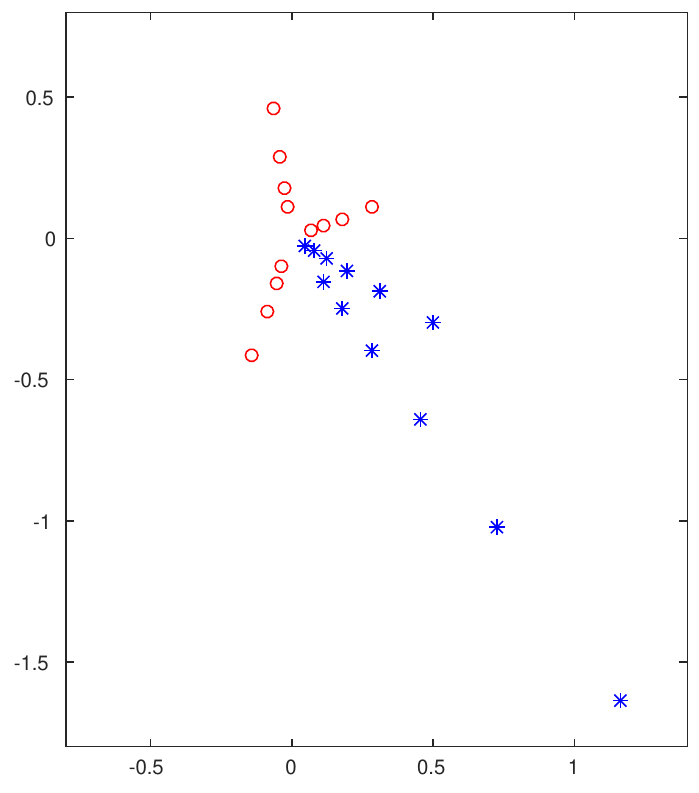}
\caption{Left: A randomly chosen frame $E$ of size $m=12$ in $\R^2$ plotted with its ordering.
Right: Two different $(\beta,m,l)$-duals of $E$ with $\beta = 1.6$. The frame ratio of the beta dual for $l=2$ (stars) is approximately 9.15 and for $l=3$ (circles) 2.07. (For comparison, the frame ratio of $E$ is equal to 1.53.)}
\label{fig:betarandplot}
\end{figure}

\subsection{Bounding the analysis distortion}
As implied by the discussion in the previous subsection, the nature of the bounds on the analysis distortion that we will be able to obtain will depend on whether we opt for $l=k$ or $l > k$. 
The next theorem is based on the former choice.

\begin{theorem}\label{thm_leqk}
Let $E$ be a standard Gaussian frame of $m$ vectors in $\R^k$, $m\geq k$. For any $L \geq 2$, $a > 0$, and $\varepsilon > 0$, we have 
\[ \sDa(E,L) <  2e (1+a) \varepsilon^{-1} m^{3/2}  L^{-\lfloor m/k \rfloor} \]
with probability at least $1- e^{-2a^2 m}- \varepsilon $.
\end{theorem}
\begin{proof}
Let $\mu = 2(1+a)\sqrt{m}$. For any $\beta\geq 1$ (whose specific value will be provided shortly), consider the event ${\mathcal M}(a)\cap \mathcal{E}(\gamma,\beta,k)$ where $\gamma = \frac{\varepsilon}{\sqrt{k}}\sigma_{\lfloor m/k \rfloor}$. By \eqref{prob-Mc} and Proposition \ref{P:squaregaussminsv}, this event occurs with probability at least $1- e^{-2a^2 m} -\varepsilon$. With the simple bound $\sigma_{\lfloor m/k \rfloor} \geq \beta^{-1}$, we have 
$ \sigmin(V_{\beta,\tilde m,k}\tilde E_k) \geq \frac{\varepsilon}{\beta \sqrt{k}}$, and therefore 
the analysis distortion bound \eqref{DaEL_bound2} with $l=k$ yields
\[ \sDa(E,L) \leq  k \varepsilon^{-1} \delta  \beta^{-\lfloor m/k \rfloor+1} \]
for any $\delta>0$ such that $(\beta,\delta) \in S_{\mu,L}$.
Now we specify the pair $(\beta,\delta) \in S_{\mu,L}$. With $\alpha =\lfloor m/k \rfloor-1$, we set 
$\beta = L(\alpha+1)/(\alpha+2)$ and $\delta = \mu(\alpha +2)/L$ as in  
Lemma \ref{calculus_lemma} and obtain
\[
\sDa(E,L) < k \varepsilon^{-1} \mu e  \lfloor m/k \rfloor L^{-\lfloor m/k \rfloor}.
\]
Substituting the value of $\mu$ and simplifying completes the proof.
\end{proof}

The parameter $a$ is largely inconsequential and could be fixed, e.g. $a=1$. Meanwhile, the most favorable choice of $\varepsilon$ depends on the probabilistic assurance as well as the analysis distortion guarantee that is desired. For example, if $k$ is comparable to $\log L$, then choosing $\epsilon = L^{-\eta  \lfloor m/k \rfloor}$ for some (small) $\eta > 0$ would guarantee, without affecting the analysis distortion significantly, that the failure probability is exponentially small in $m$. However, when $k$ is large (or $L$ is small), this choice will not correspond to a practical scenario. In any case, the above theorem is sufficient to conclude that as $m\to \infty$, we have $\sDa(E,L) =  O\left(L^{-(1+o(1))m/k}\right )$ with probability $1-o(1)$.

We now present an alternative probabilistic bound on the analysis distortion based on the choice $l > k$, which will ultimately lead to the proof of Theorem \ref{T:MAIN}. 

\begin{theorem}\label{thm_lgtrk}
Let $E$ be a standard Gaussian frame of $m$ vectors in $\R^k$, $m > k$. 
For any $L \geq 2$, $a > 0$, $\varepsilon > 0$, and $l$ such that $k < l \leq m$, we have 
\begin{equation}\label{analysis_dist_2}
\sD_{\mathrm{a}}(E,L) < 4e(1+a)\varepsilon^{-1} m^{3/2}l^{-1} L^{-\lfloor m/l \rfloor} 
\end{equation}
with probability at least 
$1- e^{-2a^2 m} -2(10+8\sqrt{\log \varepsilon^{-1}})^k  e^{l/2} \varepsilon^{l-k+1}$.
\end{theorem}
\begin{proof}
We follow the same route as in Theorem \ref{thm_leqk}.
Again, let $\mu := 2(1+a)\sqrt{m}$, and for $\beta \geq 1$ to be specified shortly, consider the event ${\mathcal M}(a)\cap \mathcal{E}(\gamma,\beta,l)$ (recall the definition given in \eqref{Egamma}) where 
$\gamma = \varepsilon \sigma_{\lfloor m/l \rfloor} \sqrt{l}/2$. 
By \eqref{prob-Mc} and  Theorem \ref{T:svmingaussrect}, this event occurs with probability at least $1- e^{-2a^2 m} -2(10+8\sqrt{\log \varepsilon^{-1}})^k  e^{l/2} \varepsilon^{l-k+1}$.
Again using $\sigma_{\lfloor m/l \rfloor} \geq \beta^{-1}$, we have 
$ \sigmin(V_{\beta,\tilde m,l}\tilde E_l) \geq \frac{\varepsilon \sqrt{l}}{2 \beta}$, and therefore 
the analysis distortion bound \eqref{DaEL_bound2} yields
\[ \sDa(E,L) \leq  2 \varepsilon^{-1} \delta  \beta^{-\lfloor m/l \rfloor+1} \]
for any $\delta>0$ such that $(\beta,\delta) \in S_{\mu,L}$.
With $\alpha =\lfloor m/l \rfloor-1$, we again set 
$\beta = L(\alpha+1)/(\alpha+2)$ and $\delta = \mu(\alpha +2)/L$ as in  
Lemma \ref{calculus_lemma} and obtain
\[
\sDa(E,L) < 2 \varepsilon^{-1} \mu e  \lfloor m/l \rfloor L^{-\lfloor m/l \rfloor}.
\]
Substituting the value of $\mu$ and simplifying completes the proof.
\end{proof}

When $l > k$ and $\varepsilon$ is small, Theorem \ref{thm_lgtrk} provides a more favorable bound on the failure probability compared to Theorem \ref{thm_leqk} because of the presence of $\varepsilon^{l-k+1}$. We now turn this observation into a concrete form by suitably choosing the values of $\varepsilon$ and $l$.

For simplicity, let $a=1$. For any (small) $\theta \in (0,1)$, let $\varepsilon = L^{-\theta m/k}$ and $l = k+ \lfloor \theta k \rfloor$.
Suppose $k \geq 1/\theta$ and $m/k \geq (1+\theta)/\theta^2$. Clearly we have $l > k$ and also
\begin{equation*}
 \left \lfloor \frac{m}{l} \right \rfloor \geq \left \lfloor \frac{m}{(1+\theta)k} \right \rfloor \geq 
\frac{m}{(1+\theta)k} - 1 \geq (1-\theta)\frac{m}{k}
\end{equation*}
so that the error bound \eqref{analysis_dist_2} is effective and can be simplified to
\begin{equation*} 
\sD_{\mathrm{a}}(E,L) < 8em^{3/2}k^{-1} L^{-(1-2\theta)\frac{m}{k}}. 
\end{equation*}
Furthermore, this bound is simultaneously valid for all $L \geq 2$ if the event
${\mathcal M}(1) \cap {\mathcal E}_\theta$ where
\begin{equation*}
{\mathcal E}_\theta := \bigcap_{L=2}^\infty {\mathcal E}(\gamma,\beta,l)
\end{equation*}
occurs.
Here $\gamma = \varepsilon \sigma_{\lfloor m/l \rfloor} \sqrt{l}/2$ and $\beta = L(\alpha+1)/(\alpha+2)$ with $\alpha = \lfloor m/l \rfloor - 1$ are as in the proof of Theorem \ref{thm_lgtrk}. 

By Theorem \ref{T:svmingaussrect}, 
the probability of $\mathcal{E}(\gamma,\beta,l)^c$ is bounded by
$2(10+8\sqrt{\log \varepsilon^{-1}})^k  e^{l/2} \varepsilon^{l-k+1}$. To simplify this bound, note that
$13 \sqrt{\log 2} > 10$ so that 
$10+8\sqrt{\log \varepsilon^{-1}} < 21\sqrt{(m/k)\log L }$. Hence, together with the observation that  $l- k+1= \lfloor \theta k \rfloor +1 > \theta k$, we have
\begin{equation*}
{\mathbb P}(\mathcal{E}(\gamma,\beta,l)^c) \leq 
2 \exp\left\{k \left(\log 21 +\frac{1}{2}  \log\frac{m}{k} + \frac{1}{2} \log \log L  + \frac{1}{2}(1+\theta)  \right)
-m \theta^2 \log L \right \}.
\end{equation*}
There exists an absolute constant $c_1 > 0$ such that if
$\frac{m}{k} \geq \frac{c_1 }{\theta^2} \log \frac{m}{k}$, then
\begin{equation*}
k \left( \log 21 +\frac{1}{2}  \log\frac{m}{k} + \frac{1}{2} \log \log L  + \frac{1}{2}(1+\theta)  \right)
\leq \frac{1}{2}m \theta^2 \log L
\end{equation*}
so that 
\begin{equation*}
\mathbb{P}(\mathcal{E}(\gamma,\beta,l)^c) \leq 2 L^{-\frac{1}{2} \theta^2 m}. 
\end{equation*}
By readjusting the value of $c_1$ if necessary, we are guaranteed that 
$m \geq 4/\theta^2$ so that $\frac{1}{2}\theta^2 m \geq 2$ and
\begin{equation*}
\mathbb{P} ({\mathcal E}_\theta^c) \leq 
\sum_{L=2}^\infty 2 L^{-\frac{1}{2}\theta^2 m} <
2\left( 2^{-\frac{1}{2}\theta^2 m} + \int_2^\infty t^{-\frac{1}{2}\theta^2 m}\,\mathrm{d}t \right) \leq
6\cdot2^{-\frac{1}{2}\theta^2 m}. 
\end{equation*}
Setting $\eta = 2\theta$ and restricting $\theta \in (0,1/2)$, we obtain the following corollary: 
\begin{corollary}\label{cor}
There exists an absolute constant $c_1 > 0$ such that if $\eta \in (0,1)$, and 
$m$ and $k$ are such that $k \geq 2/\eta$ and $\frac{m}{k} \geq \frac{c_1}{\eta^2} \log \frac{m}{k}$, a standard Gaussian frame $E$ of $m$ vectors in $\R^k$ satisfies 
\[ \sD_{\mathrm{a}}(E,L) < 8em^{3/2}k^{-1} L^{-(1-\eta) m/k} \]
for all $L \geq 2$ with probability at least $1 -6\cdot 2^{-\frac{1}{8}\eta^2 m} - e^{-2m}$.
\end{corollary}

\paragraph{Proof of Theorem \ref{T:MAIN}.}
We make the following observations. First, increasing the value of $c_1$ in Corollary \ref{cor} if necessary, the assumption
$\frac{m}{k} \geq \frac{c_1}{\eta^2} \log \frac{m}{k}$
implies that $\frac{m}{k}  \geq \frac{c_1}{\eta^2}$, and therefore
\begin{equation*}
8e \left (\frac{m}{k} \right)^{3/2} <  e^{4+1.5 \log(m/k)} < L^{c_3\eta^2 m/k},
\end{equation*}
where $c_3 = 8/c_1$. Again increasing the value of $c_1$ if necessary, 
we can assume that $c_3\leq 1$ so that $1-\eta - c_3  \eta^2 \geq 1-2\eta$. Next, we can restrict $\eta \in (0,1/2)$, set $\eta'=2\eta$ and $c_1'=4c_1$. This yields the error bound $\sqrt{k} L^{-(1-\eta')m/k}$. For the probability guarantee, 
note that $m = k(m/k)\geq (4/\eta') (c_1'/\eta'^2) = 4c_1'/\eta'^3$. Once again increasing the value of $c_1'$ if necessary, we can find an absolute constant $c_2>0$ such that $6\cdot 2^{-\frac{1}{32}\eta'^2 m} + e^{-2m}
\leq e^{-c_2 \eta'^2 m}$ for all admissible $m$, hence Theorem \ref{T:MAIN} follows once we denote $\eta'$ by $\eta$ and $c_1'$ by $c_1$ again.

\subsection{Results for sub-Gaussian distributions}

Our methods are applicable in the case of any other probability distribution governing the frame $E$ since the basic quantization algorithm and the resulting bound on the analysis distortion are deterministic. As we discussed in depth for the Gaussian distribution, two types of probabilistic bounds need to be established for $E$: The first is an upper bound on $\|E\|_{2 \to \infty}$ which may be replaced by an upper bound on $\sigma_\mathrm{max}(E)$, and the second is a lower bound on $\sigma_\mathrm{min}(V_{\beta,\tilde m,l}\tilde E_l)$. 

For simplicity we will continue to assume that the entries of $E$ are chosen independently, but now more generally from a centered sub-Gaussian distribution of unit variance. Recall that a random variable $X$ on $\R$ is called sub-Gaussian if there exists $K > 0$, called the {\em sub-Gaussian moment} of $X$, such that
${\mathbb P}(|X| >t) \leq 2 e^{-t^2/K^2}$ for all $t> 0$ \cite{vershynin2010introduction}.

Regarding the upper bound on $\sigma_\mathrm{max}(E)$, the generalization of Proposition \ref{ds2001} to sub-Gaussian random matrices says that 
\begin{equation*}
{\mathbb P}\Big (\sigma_\mathrm{max}(E) > K_1(\sqrt{m} + \sqrt{k}) +t \Big)
\leq 2 e^{-t^2/K_2^2}
\end{equation*}
for positive absolute numerical constants $K_1$ and $K_2$ (see, e.g. \cite{rudelson2010non}). Consequently, we may simply set $\mu:=2(K_1 +1)\sqrt{m}$ which is essentially the same as in the Gaussian case, albeit for a slightly larger constant. In the special case of bounded distributions, then it is actually possible to set $\mu:=O(\sqrt{k})$ which is even better than the Gaussian case, and in a certain sense, the best we can ever achieve.

As for the lower bound on $\sigma_\mathrm{min}(V_{\beta,\tilde m,l}\tilde E_l)$, note that the entries of $\Omega := V_{\beta,\tilde m,l}\tilde E_l$ are also independent because each entry of $E$ influences at most one entry of $ V_{\beta,\tilde m,l}\tilde E_l$, thanks to the structure of the beta condensation matrix $V_{\beta,\tilde m,l}$. (Again thanks to this structure, one may even weaken the entry-wise independence assumption we made above.) In addition, the entries of $\Omega$ are identically distributed and are centered sub-Gaussian random variables \cite[Lemma 5.9]{vershynin2010introduction}. Proposition \ref{P:squaregaussminsv} and Theorem \ref{T:svmingaussrect} have also been generalized to the sub-Gaussian case, and in fact as part of one common result:
If $\Omega$ is an $l\times k$ matrix whose entries are independent and identically distributed sub-Gaussian random variables with zero mean and unit variance, then 
\begin{equation}\label{subGaussiansmallestsingval}
{\mathbb P}\Big (\sigma_\mathrm{min}(\Omega) \leq \varepsilon\big (\sqrt{l} - \sqrt{k-1}\big ) \Big) \leq (K_3\varepsilon)^{l-k+1} + c^l
\end{equation}
where $K_3 > 0$ and $c \in (0,1)$ depend only on the sub-Gaussian moment of the entries \cite{rudelson2009smallest}. 

With these tools, it is possible to obtain fairly strong bounds on the analysis distortion of sub-Gaussian frames, but not quite as strong as those obtainable for the Gaussian case. The main problem stems from the presence of the $c^l$ term in \eqref{subGaussiansmallestsingval} which, as explained in \cite{rudelson2009smallest,rudelson2010non}, is characteristic of the sub-Gaussian case. As such, this term prevents us from stating an exact analog of Theorem \ref{T:MAIN}. In particular, $l$ must be required to go to infinity in order to ensure that the failure probability bound vanishes as $m \to \infty$. Consequently, the resulting analysis distortion bound would not be near-optimal in the sense we currently have, but of the form $L^{-o(m)}$. For example, if we let $l \approx m^\kappa k^{1-\kappa}$ for some $\kappa \in (0,1)$, then we can obtain a ``root-exponential'' distortion bound of the form $L^{-(m/k)^{1-\kappa}}$ which is guaranteed also up to root-exponentially small failure probability. Alternatively, the more conservative choice $l \approx k+\log m$ would result in a distortion bound of the form $L^{-m/(k+\log m)}$ up to a failure probability of $O(m^{-a})$ where $c=e^{-a}$.

\section*{Acknowledgements}

The authors would like thank Thao Nguyen, Rayan Saab and \"Ozg\"ur Y\i lmaz for the useful conversations on the topic of this paper and the anonymous referees for the valuable comments and the references they have brought to our attention.

\appendix

\section{Appendix}\label{appendix}
\begin{lemma} \label{P:gaussbdlow}
Let $\xi \sim \mathcal{N}(0,I_l)$. For any $0 < \varepsilon \leq 1$, we have
\begin{equation*}
\mathbb{P}\left (\| \xi\|_2 \leq \varepsilon \sqrt{l}  \right)
\leq \varepsilon^{l} e^{(1-\varepsilon^2)l/2}.
\end{equation*}
\end{lemma}
\begin{proof}
For any $t \geq 0$, we have 
\begin{equation*}
\mathbb{P}\left (\| \xi\|^2_2 \leq \varepsilon^2 l  \right)
\leq \int_{\R^l} e^{(\varepsilon^2 l - \|x\|_2^2)t/2} \frac{e^{-\|x\|_2^2/2}}{(2\pi)^{l/2}} \,\mathrm{d}x
= e^{\varepsilon^2 t l/2} \int_{\R^l} \frac{e^{-(1+t)\|x\|_2^2/2}}{(2\pi)^{l/2}} \,\mathrm{d}x
= \left( \frac{e^{\varepsilon^2 t}}{1+t} \right)^{l/2}.
\end{equation*}
Choosing $t = \varepsilon^{-2}-1$ yields the desired bound.
\end{proof}

\paragraph{Proof of Theorem \ref{T:svmingaussrect}.}
For an arbitrary $\tau > 1$, let $\mathcal{E}_1$ be the event $\{ \|\Omega \|_{2\to 2} \leq 2\tau \sqrt{l}\}$.
Proposition \ref{ds2001} with $m=l$, $E=\Omega$, $t=2(\tau-1)\sqrt{l}$ implies 
\begin{equation*}
\mathbb{P}(\mathcal{E}_1^c) \leq e^{-2(\tau-1)^2 l}.
\end{equation*}
Next, consider a $\rho$-net $Q$ of the unit sphere of $\R^k$ with 
$|Q| \leq 2k(1+2/\rho)^{k-1}$ (see \cite[Proposition 2.1]{rudelson2009smallest}), where we set $\rho= \varepsilon/(4\tau)$.
Let $\mathcal{E}_2$ be the event 
$\{ \|\Omega z\|_2 \geq \varepsilon \sqrt{l},\ \forall z \in Q\}$.
For each fixed $z\in \R^k$ with unit norm, $\Omega z$ 
has entries that are i.i.d. $\mathcal{N}(0,1)$.
By Lemma \ref{P:gaussbdlow}, we have
\begin{equation*}
\mathbb{P}(\mathcal{E}_2^c) \leq |Q| \varepsilon^{l} e^{(1-\varepsilon^2)l/2}
\leq 2k(\varepsilon + 8 \tau)^{k-1} \varepsilon^{l-k+1} e^{(1-\varepsilon^2)l/2}.
\end{equation*}

Suppose the event $\mathcal{E}_1 \cap \mathcal{E}_2$ occurs. 
For any unit norm $x \in \R^k$, 
there exists a $z \in Q$ such that $\|x-z\|_2 \leq \rho$. 
Then $\| \Omega(x-z)\|_2 \leq 2\tau \rho \sqrt{l} = \varepsilon\sqrt{l}/2$ and $\|\Omega z\|_2 \geq \varepsilon \sqrt{l}$, so that 
\begin{equation*}
\|\Omega x\|_2 \geq \|\Omega z\|_2 - \|\Omega(x-z)\|_2 \geq \varepsilon \sqrt{l}/2,
\end{equation*}
hence $\sigmin(\Omega) \geq \varepsilon \sqrt{l}/2$. It follows that 
$\mathcal{F}:=\left \{\sigmin(\Omega) \leq \varepsilon \sqrt{l}/2 \right\} \subset 
\mathcal{E}_1^c \cup \mathcal{E}_2^c$, and therefore 
\begin{equation*}
 \mathbb{P}\left(\mathcal{F}\right) 
\leq e^{-2(\tau-1)^2 l} + 2k(\varepsilon + 8 \tau)^{k-1} \varepsilon^{l-k+1} e^{l/2}. 
\end{equation*}

We still have the freedom to choose $\tau > 1$ as a function of $\varepsilon$, $l$, and $k$.
For simplicity, we choose $\tau = 1+\sqrt{\log \varepsilon^{-1}}$ so that 
$e^{-2(\tau-1)^2 l} = \varepsilon^{2l}$. Noting that 
$1 + k(1+8\tau)^{k-1} < (2+8\tau)^k$, we obtain
\begin{equation*} 
\mathbb{P}\left(\mathcal{F}\right) 
<  \varepsilon^{l-k+1} \left(1 + 2k(1 + 8\tau)^{k-1}  e^{l/2}\right)
<  2 \left (10+8\sqrt{\log \varepsilon^{-1}}\right)^k  e^{l/2} \varepsilon^{l-k+1}. 
\end{equation*}

\bibliographystyle{plain}
\bibliography{betadual}

\begin{thebibliography}{10}

\bibitem{BFNPW14}
R.G. Baraniuk, S.~Foucart, D.~Needell, Y.~Plan, and M.~Wootters.
\newblock Exponential decay of reconstruction error from binary measurements of
  sparse signals.
\newblock {\em CoRR}, abs/1407.8246, 2014.

\bibitem{BPY2}
J.J. Benedetto, A.M. Powell, and {\"O}.~Y{\i}lmaz.
\newblock Second-order sigma-delta {$(\Sigma\Delta)$} quantization of finite
  frame expansions.
\newblock {\em Appl. Comput. Harmon. Anal.}, 20(1):126--148, 2006.

\bibitem{benedetto2006sigma}
J.J. Benedetto, A.M. Powell, and \"O. Y{\i}lmaz.
\newblock Sigma-delta quantization and finite frames.
\newblock {\em Information Theory, IEEE Transactions on}, 52(5):1990--2005,
  2006.

\bibitem{blum2010sobolev}
J.~Blum, M.~Lammers, A.M. Powell, and {\"O}.~Y{\i}lmaz.
\newblock Sobolev duals in frame theory and sigma-delta quantization.
\newblock {\em Journal of Fourier Analysis and Applications}, 16(3):365--381,
  2010.

\bibitem{framepath2007}
B.G. Bodmann and V.I. Paulsen.
\newblock Frame paths and error bounds for sigma-delta quantization.
\newblock {\em Appl. Comput. Harmon. Anal.}, 22(2):176--197, 2007.

\bibitem{chazelle2002discrepancy}
B.~Chazelle.
\newblock {\em The discrepancy method: randomness and complexity}.
\newblock Cambridge University Press, 2002.

\bibitem{Dajani}
K.~Dajani and C.~Kraaikamp.
\newblock {\em Ergodic theory of numbers}, volume~29 of {\em Carus Mathematical
  Monographs}.
\newblock Mathematical Association of America, Washington, DC, 2002.

\bibitem{DDGV}
I.~Daubechies, R.A. DeVore, C.S. G{\"u}nt{\"u}rk, and V.A. Vaishampayan.
\newblock A/{D} conversion with imperfect quantizers.
\newblock {\em IEEE Trans. Inform. Theory}, 52(3):874--885, 2006.

\bibitem{ds2001}
K.R. Davidson and S.J. Szarek.
\newblock Local operator theory, random matrices and {B}anach spaces.
\newblock In {\em Handbook of the geometry of {B}anach spaces, {V}ol. {I}},
  pages 317--366. North-Holland, Amsterdam, 2001.

\bibitem{ds2003}
K.R. Davidson and S.J. Szarek.
\newblock Addenda and corrigenda to: ``{L}ocal operator theory, random matrices
  and {B}anach spaces''.
\newblock In {\em Handbook of the geometry of {B}anach spaces, {V}ol.\ 2},
  pages 1819--1820. North-Holland, Amsterdam, 2003.

\bibitem{deift2011optimal}
P.~Deift, F.~Krahmer, and C.S. G{\"u}nt{\"u}rk.
\newblock An optimal family of exponentially accurate one-bit sigma-delta
  quantization schemes.
\newblock {\em Communications on Pure and Applied Mathematics}, 64(7):883--919,
  2011.

\bibitem{derpich2008}
M.S. Derpich, E.I. Silva, D.E. Quevedo, and G.C. Goodwin.
\newblock On optimal perfect reconstruction feedback quantizers.
\newblock {\em IEEE Trans. Signal Process.}, 56(8, part 2):3871--3890, 2008.

\bibitem{edelman1988eigenvalues}
A.~Edelman.
\newblock Eigenvalues and condition numbers of random matrices.
\newblock {\em SIAM Journal on Matrix Analysis and Applications},
  9(4):543--560, 1988.

\bibitem{GJ}
I.~Galton and H.T. Jensen.
\newblock Oversampling parallel delta-sigma modulator {A}/{D} conversion.
\newblock {\em Circuits and Systems II: Analog and Digital Signal Processing,
  IEEE Transactions on}, 43(12):801--810, Dec 1996.

\bibitem{goyal1998quantized}
V.K. Goyal, M.~Vetterli, and N.T. Thao.
\newblock Quantized overcomplete expansions in {$\mathbb{R}^N$}: analysis,
  synthesis, and algorithms.
\newblock {\em Information Theory, IEEE Transactions on}, 44(1):16--31, 1998.

\bibitem{gunturk2010sobolev}
C.~S. G{\"u}nt{\"u}rk, M.~Lammers, A.~M. Powell, R.~Saab, and {\"O}.~Y{\i}lmaz.
\newblock Sobolev duals for random frames and {$\Sigma\Delta$} quantization of
  compressed sensing measurements.
\newblock {\em Found. Comput. Math.}, 13(1):1--36, 2013.

\bibitem{gunturk2003one}
C.S. G{\"u}nt{\"u}rk.
\newblock One-bit sigma-delta quantization with exponential accuracy.
\newblock {\em Communications on Pure and Applied Mathematics},
  56(11):1608--1630, 2003.

\bibitem{MathAtoD}
C.S. G{\"u}nt{\"u}rk.
\newblock Mathematics of analog-to-digital conversion.
\newblock {\em Comm. Pure Appl. Math.}, 65(12):1671--1696, 2012.

\bibitem{IwenSaab}
M.~Iwen and R.~Saab.
\newblock Near-optimal encoding for sigma-delta quantization of finite frame
  expansions.
\newblock {\em Journal of Fourier Analysis and Applications}, 19(6):1255--1273,
  2013.

\bibitem{krahmer2013subGaussian}
F.~Krahmer, R.~Saab, and {\"O}.~Y{\i}lmaz.
\newblock Sigma-delta quantization of sub-gaussian frame expansions and its
  application to compressed sensing.
\newblock {\em Information and Inference}, 3(1):40--58, 2014.

\bibitem{alternative2010}
M.~Lammers, A.M. Powell, and {\"O}.~Y{\i}lmaz.
\newblock Alternative dual frames for digital-to-analog conversion in
  sigma-delta quantization.
\newblock {\em Adv. Comput. Math.}, 32(1):73--102, 2010.

\bibitem{LPY}
M.C. Lammers, A.M. Powell, and {\"O}.~Y{\i}lmaz.
\newblock On quantization of finite frame expansions: sigma-delta schemes of
  arbitrary order.
\newblock {\em Proc. SPIE 6701, Wavelets XII, 670108}, 6701:670108--670108--9,
  2007.

\bibitem{matousek_discrepancy}
J.~Matou{\v{s}}ek.
\newblock {\em Geometric discrepancy}, volume~18 of {\em Algorithms and
  Combinatorics}.
\newblock Springer-Verlag, Berlin, 1999.
\newblock An illustrated guide.

\bibitem{molino2012approximation}
V.~Molino.
\newblock {\em Approximation by Quantized Sums}.
\newblock PhD thesis, New York University, 2012.

\bibitem{Parry}
W.~Parry.
\newblock On the {$\beta $}-expansions of real numbers.
\newblock {\em Acta Math. Acad. Sci. Hungar.}, 11:401--416, 1960.

\bibitem{powell2013quantization}
A.M. Powell, R.~Saab, and {\"O}.~Y{\i}lmaz.
\newblock Quantization and finite frames.
\newblock In {\em Finite Frames}, pages 267--302. Springer, 2013.

\bibitem{rudelson2009smallest}
M.~Rudelson and R.~Vershynin.
\newblock Smallest singular value of a random rectangular matrix.
\newblock {\em Communications on Pure and Applied Mathematics},
  62(12):1707--1739, 2009.

\bibitem{rudelson2010non}
M.~Rudelson and R.~Vershynin.
\newblock Non-asymptotic theory of random matrices: extreme singular values.
\newblock In {\em Proceedings of the {I}nternational {C}ongress of
  {M}athematicians. {V}olume {III}}, pages 1576--1602. Hindustan Book Agency,
  New Delhi, 2010.

\bibitem{thao1996lower}
N.T. Thao and M.~Vetterli.
\newblock Lower bound on the mean-squared error in oversampled quantization of
  periodic signals using vector quantization analysis.
\newblock {\em Information Theory, IEEE Transactions on}, 42(2):469--479, 1996.

\bibitem{vershynin2010introduction}
R.~Vershynin.
\newblock Introduction to the non-asymptotic analysis of random matrices.
\newblock In {\em Compressed sensing}, pages 210--268. Cambridge Univ. Press,
  Cambridge, 2012.

\end{thebibliography}

\end{document}